\documentclass[11pt]{article}

\usepackage{mpk-notes}
\usepackage{tabularx}
\newcolumntype{C}[1]{>{\centering\arraybackslash}p{#1}}

\title{Who Can Win a Single-Elimination Tournament?}
\author{Michael P. Kim, Warut Suksompong, and Virginia Vassilevska Williams}
\date{\today}

\begin{document}

\maketitle

\abstract{
  A single-elimination (SE) tournament is a popular way to select a winner in both sports competitions and in elections. A natural and well-studied question  is  
  the tournament fixing problem (TFP): given the set of all pairwise match outcomes, can a tournament organizer rig an SE tournament by adjusting the initial seeding so that their favorite player wins?
   We prove new sufficient conditions on the pairwise match outcome information and the favorite player, under which there is guaranteed to be a seeding where
   the player wins the tournament. Our results greatly generalize previous results.
    We also investigate the relationship between the set of players that can win an SE tournament under some seeding (so called SE winners)
    and other traditional tournament solutions. In addition, we generalize and strengthen prior work on probabilistic models for generating tournaments. For instance, we show that {\em every}
    player in an $n$ player tournament generated by 
    the Condorcet Random Model 
    will be an SE winner even when the noise is as small as possible, $p=\Theta(\ln n/n)$; prior work only had such results for $p\geq \Omega(\sqrt{\ln n/n})$.
    We also establish new results for significantly more general generative models.
}

\section{Introduction}
\label{sec:intro}

A single-elimination (SE) tournament, also known as a
binary-cup election, is a popular way to select a winner among multiple candidates/players.
In an SE tournament, pairs of players are matched according to
an initial seeding, the winners of these pairs advance to the
next round, and the losers are eliminated after a single loss.
Play continues according to the seeding until a single player, the winner,
remains.
SE tournaments are popular in sports competitions, both among fans due to their
exciting ``do-or-die" nature, and
among tournament organizers due to their efficiency.
In contrast with a round-robin tournament,
which requires $\Theta(n^2)$ matches to be played between $n$ players,
the winner of an SE tournament
is decided after a total of $n-1$ matches.  In tournaments
like the NCAA March Madness or the US Open Tennis Championships,
involving more than $64$ teams each,
the difference between a linear and quadratic number of matches
is quite pronounced.

While the efficiency of SE tournaments is desirable,
the winner of a given SE tournament can depend significantly
on the initial seeding.  A series of works \cite{lang07,HDW07,HDKW08,Vu09,virgiaaai,tournwine,tournijcai,mattei12,haziz,kvijcai15}
have investigated
how easily the winner of SE tournaments can be manipulated
simply by adjusting the seeding of the tournament.
Formally, the problem is called the \emph{tournament fixing problem} (TFP), or the \emph{agenda control problem for balanced knockout tournaments}.
In TFP, we are given a set of players $V$, information
for each pair of players $(u,w)$ about whether $u$ or
$w$ would win in a head-to-head matchup, and a player of
interest $v$; then, we are asked the following question:
is there a seeding to a balanced SE tournament where $v$ wins?
TFP is known to be $\NP$-Hard \cite{haziz} with the best-known
algorithm running in $2^n\poly(n)$ time \cite{kvijcai15}.
Thus, unless $\P = \NP$, in general, it is intractable
to determine which players can win an SE tournament.
Nevertheless, a number of works on TFP have produced
``structural results," which argue that for certain classes of
instances, one can find a winning seeding for $v$ in polynomial (and often linear) time
\cite{virgiaaai,tournijcai,kvijcai15}.
These structural results suggest that in many practical settings,
the winner of an SE tournament is susceptible to manipulation,
because many players have winning seedings that can be found
efficiently.
Furthermore, under reasonable probabilistic models for generating
tournaments, these structural results have been shown to occur
with high probability \cite{virgiaaai,tournwine}, further suggesting that
the worst-case hardness results may not apply to real-world
instances.
In other words, in many actual tournaments, it is completely
feasible for SE tournament organizers to rig the outcome
of the competition.\footnote{While the set of games to be played in a round-robin tournament are fixed, the organizers can still affect the fairness of the tournament by adjusting the schedule (see, e.g., \cite{Suksompong16}).} Experimental results \cite{russellthesis} investigate this finding in practical settings.

While TFP can be seen as a way to understand manipulation in
competition and elections,
studying conditions under which players can and cannot win SE
tournaments can also be seen as part of a larger study of
{\em tournament solutions}: different ways to define the winners of a round-robin tournament.
The input to TFP can be viewed as a \emph{tournament} 
$T = (V,E)$, or a complete, oriented graph where for all pairs of nodes
$u,w \in V$, exactly one of $(u,w)$ and $(w,u)$ is an element of
$E$; $u$ points to $w$ if $u$ would win in the match between
$u$ and $w$.
The study of tournaments is central to social choice theory;
they provide a general framework for representing
the outcomes between players in a round-robin tournament,
or more generally, pairwise preferences between alternatives, often generated from voter information.
As such, an essential question of social choice theory asks:
given a tournament, how should we select a set of winners?
SE tournaments provide one way of answering this question;
we say that a player $v \in V$ is an \emph{SE winner} if
there is some seeding, under which $v$ wins the resulting
SE tournament.  The study of tournament solutions includes many well-studied other concepts (see e.g. \cite{Laslier97,BrandtEtAlChapter}).
One classical example is
the Copeland
set, consisting of the players with the maximum number of wins
in the tournament.  
A natural question to investigate is how these traditional notions
of strength in round-robin tournaments relate to the notion of
strength in an SE tournament.

\paragraph{Results}
In this work, we improve our understanding of conditions
on the input tournament and player of interest that are
sufficient for the player to be an SE winner.
Many previous structural results involve the notion of a \emph{king},
or a player $v$ where for every other player
$u \in V \setminus \set{u}$, $v$ either beats $u$ directly,
or $v$ beats some $w$ who beats $u$.
We present a vast generalization of many of the known structural
results involving kings, showing that essentially any ``combination"
of the known sufficient conditions for a king to be an SE winner
is also sufficient for the king to be a winner.

In particular, recall
the following structural results from \cite{virgiaaai},
where given a tournament $T$ and a player $v$, we can find
a winning seeding for $v$ in polynomial time.
One class of tractable instances are those where
every player $w$, who beats $v$, wins against at most
as many players as $v$ beats.
It is also known that if $v$ is a king and wins
against more than half the players or is a
``superking" and every $w$ whom $v$ beats indirectly loses to
at least $\log{n}$ players whom $v$ beats directly,
then $v$ will be able to win an SE tournament.
While these results have been useful on their own for
showing that tournaments generated by certain random models
are likely to have many players who can win \cite{virgiaaai,tournwine},
it is natural to wonder how robust these results are to
changes in the exact sufficient conditions.
Recent results of \cite{kvijcai15} seem to suggest that
 the parameters for these structural results
are brittle; namely, when the exact parameters of the conditions
are relaxed, finding a winning seeding for $v$ (if it exists)
becomes $\NP$-Hard. In Theorem~\ref{thm:generalizedstructural},
we provide a broad generalization of the
three structural results stated above.  We show that
these conditions are actually flexible in the sense that
if the players who beat some king $v$, can be partitioned into
groups that satisfy these sufficient conditions,
then $v$ can win an SE tournament.
Additionally, we extend the work on \emph{$3$-kings}
(or players who have win-distance $\le 3$ to every other player),
introduced in \cite{kvijcai15}, and give a new set of
sufficient conditions for a $3$-king win an SE tournament.

In Section~\ref{sec:solutions}, we are able to apply these and
other known structural results to understand the relationship
between SE winners and the winners according to other tournament
solutions.  In particular, Theorem~\ref{thm:solutions}
shows that the players selected by a number of well-studied
tournament solutions are also SE winners, including the
Copeland set described above.  Another class of tournament solutions of
interest was introduced in \cite{Laslier97} as a natural extension
of the Copeland set.  In these ``iterative matrix solutions,"
we consider the tournament matrix $A$ (corresponding to the
adjacency matrix of the underlying tournament graph);
a player is included in the $k$th iterative matrix solution,
if they have the maximum number of ``wins" in $A^k$.
We give a new interpretation of this solution and use it
to show that for sufficiently large tournaments, the players
in the iterative matrix solutions will also be SE winners.

Finally, in Section~\ref{sec:probabilistic}, we investigate
probabilistic models for generating random tournaments, and
the resulting structure of such tournaments.  In particular,
we start by giving an improved result for tournaments generated
by the Condorcet Random (CR) Model.  The CR Model assumes an
underlying order to players, where stronger players generally
win against weaker players and are only upset with some small
probability $p$. We demonstrate that in tournaments generated by
the CR Model, even when the probability of upsets $p$ is
$\Theta(\ln{n}/n)$, with high probability every player in the
tournament will have a winning seeding that can be discovered
efficiently.  This upset rate $p$ is optimal (up to constant
factors) because a player needs to win $\log{n}$ matches in order
to win an SE tournament. Our result greatly improves on the previous
best result from \cite{virgiaaai}, which proves an analogous
claim for $p \ge \Omega(\sqrt{\ln{n}/n})$.
In light of this optimal result for the CR Model, we propose
a new generative model for tournaments that aims to remove
the structure that arises from assuming an underlying order
of players and a consistent noise parameter.
Despite the fact that the model may produce tournaments
with largely arbitrary structure, we are able to prove a
nontrivial result similar to the previous results on the CR Model.
The details of the model and our theorem statement are given
in Section~\ref{sec:probabilistic}.

All of our results are constructive.  In particular,
we demonstrate that certain conditions are sufficient for
a player $v$ to be an SE winner by giving algorithms, running
in polynomial time, that outputs a seeding where $v$ will
win.

\paragraph{Preliminaries and Notation}
We will assume throughout that all SE tournaments are
balanced, and played amongst a power of two, $n = 2^k$
for some $k \ge 0$, players.
Table~\ref{notation} provides a summary of the notation that
is used to refer to players and their neighborhood in the
underlying tournament.  For subsets $A,B \subseteq V$,
we say that $A$ dominates $B$, denoted $A \succ B$, if for all
$a \in A$ and all $b \in B$, $(a,b) \in E$.  We will abuse this
notation slightly, allowing individual players, rather than subsets,
to be related to other players or subsets.

\begin{table}
\centering
\begin{tabular}{|C{0.5\textwidth}|}

\hline
{\bf Notation}\\
\hline
$N_{out}(v) = \set{u : (v,u) \in E}$,\\
$N_{in}(v) = \set{u : (u,v) \in E}$\\
\hline
$out(v) = \card{N_{out}(v)},\
out_S(v) = \card{N_{out}(v) \cap S}$\\
\hline
$in(v) = \card{N_{in}(v)},\
in_S(v) = \card{N_{in}(v) \cap S}$\\
\hline
\end{tabular}
\caption{Summary of the notation used in this paper.}\label{notation}
\end{table}

Recall that we can define the notions of king and $3$-king of
a tournament in terms of the underlying tournament graph.
A \emph{king} is a player $v$ who has distance at most
$2$ to every other player $u \in V \setminus \set{v}$.
A \emph{$3$-king} is
the generalization of kings to players who have distance at most
$3$ to every other player.

In Section~\ref{sec:solutions}, a number of tournament solutions
are referenced.  Here, we provide brief descriptions of these
solutions; for complete definitions, we refer the interested
reader to \cite{BrandtEtAlChapter}.
The \emph{uncovered set} refers to the set of kings
in the tournament.
The \emph{Copeland set} is the set of players
of maximum out-degree in the tournament.

We say that a tournament is \emph{transitive} if we can label
the players with labels from $\set{1,\hdots,n}$
such that $\forall i,j\ i<j$
implies $i \succ j$.  Given a tournament $T$, consider flipping
edges in $T$ to produce a transitive tournament $T'$, while minimizing
the number of edges flipped. The \emph{Slater set} of $T$ is the set of
players who can be labeled $1$ (i.e., the Condorcet winner)
in such a $T'$.

The \emph{Markov set}
can be thought of as the set of players who win the most matches, in
expectation, in a ``winner-stays" tournament, where play proceeds by
repeatedly selecting a random player to play the previous winner.
This is equivalent to
finding the players of maximum probability on a random walk on the
tournament, where the graph Laplacian defines the transition matrix.

The \emph{bipartisan set} is the support of the maximal lottery
for the tournament, i.e., the support of the Nash equilibrium of the 
symmetric zero-sum game formed by the tournament matrix.

\section{Structural Results}
\label{sec:structures}

Various results are known about conditions under which a player is guaranteed to be an SE winner \cite{virgiaaai,tournijcai,kvijcai15}. Many of these results concern players who are kings. In particular, \cite{virgiaaai} showed that a ``superking'' -- a king $v$
where every player in $N_{in}(v)$ loses to at least $\log{n}$ players
from $N_{out}(v)$ -- is always an SE winner.
\cite{tournwine} showed a generalization they call a
``king of high out-degree'' -- that is, a king with out-degree
$k$, who loses to fewer than $k$ players that have out-degree greater
than $k$ -- is always an SE winner.  This result was the first to
generalize the conditions on players who can win SE tournaments.
In this section, we further generalize these results by combining their respective conditions. Moreover, we further explore the notion of 3-kings that was considered by \cite{kvijcai15} and present an alternative condition under which a 3-king can win an SE tournament.

\begin{theorem}
\label{thm:generalizedstructural}
Consider a tournament $T=(V,E)$ where $\mathcal{K}\in V$ is a king. Let $A=N_{out}(\mathcal{K})$ and $B=V\backslash(A\cup\{\mathcal{K}\})=N_{in}(\mathcal{K})$. Suppose that $B$ is a disjoint union of three (possibly empty) sets $H,I,J$ such that 
\begin{enumerate}
\item $|H|<|A|$
\item $in_A(i)\geq\log \card{V}$ for all $i\in I$ (i.e., $out_A(i)\leq|A|-\log \card{V}$ for all $i\in I$)
\item $out(j)\leq|A|$ for all $j\in J$. 
\end{enumerate}
Then $\mathcal{K}$ is an SE winner, and we can compute a winning seeding for $\mathcal{K}$ in polynomial time.
\end{theorem}

Note that the superking result \cite{virgiaaai} corresponds to the special case where $H=J=\emptyset$, while the ``king of high out-degree'' result \cite{tournwine} corresponds to the special case where $I=\emptyset$.

\begin{proof}
We proceed by induction, arguing that we can construct a seeding where, in each round, the three properties listed above and the fact that $\mathcal{K}$ is a king are maintained as invariants.
We will first take care of the cases where the tournament is small. If $|V|=1$ or 2, $B$ is empty and the result is clear.

Suppose that $|V|=4$. If $|A|\geq 2$, the result follows from \cite{tournwine}. Otherwise $|A|=1$, and $H=I=\emptyset$ and $|J|\leq 1$, which contradicts $|V|=4$. 

Suppose now that $|V|\geq 8$. If $|A|\leq 2$, then $|H|\leq 1$, $I=\emptyset$, and $|J|\leq 3$, which contradicts $|V|\geq 8$. If $I=\emptyset$, $H\cup J=\emptyset$,  or $|A|\geq|V|/2$, the result follows from \cite{tournwine} and \cite{virgiaaai}. Hence we may assume from now on that $|V|\geq 8$, $3\leq|A|<|V|/2$, $I\neq\emptyset$, and $H\cup J\neq\emptyset$. 

We will present an algorithm to compute a winning seeding for $\mathcal{K}$. The algorithm will match the players for the first round of the tournament in such a way that the leftover tournament after the first round also satisfies the conditions of the theorem. The description of the algorithm is as follows.

\begin{enumerate}
\item Perform a maximal matching $M_1$ from $A$ to $H$. 

\item Since $|H|<|A|$, we have $A\backslash M_1\neq\emptyset$. Perform a maximal matching $M_2$ (which might be an empty matching) from $A\backslash M_1$ onto $I\cup J$. 

\item If $A$ was not fully used in the two matchings, match an arbitrary unused player in $A$ with $\mathcal{K}$. Else, choose an arbitrary player $a\in A\cap M_2$ and rematch it to $\mathcal{K}$. 

\item Perform arbitrary matchings within $A,H$, and $I\cup J$. 

\item If there are leftover players, there must be exactly two of them; match them to each other.
\end{enumerate}

We prove the correctness of the algorithm by showing that the four invariants are maintained by the algorithm. Let $V',A',B',H',I',J'$ denote the subsets of $V,A,B,H,I,J$ that remain after the iteration.

\begin{enumerate}
\item $|H'|<|A'|$. We will show that $|H'|\leq|H|/2$ and $|A'|\geq|A|/2$. The claim follows since $|H|<|A|$. If $H=\emptyset$, then $|H'|<|A'|$ holds trivially, so we may assume that $H$ is nonempty. At least one player in $H$ is used in the matching $M_1$, so we have $|H'|\leq|H|/2$. We will show that the matching $M_1\cup M_2$ consists of at least two pairs. Since there can be at most two pairs in the matching provided by the algorithm in which a player in $A$ is beaten by a player outside of $A$ (i.e., the pair in which a player in $A$ is matched to $\mathcal{K}$ and the pair in which a player in $A$ is matched in the final step of the algorithm for leftover nodes), it will follow that $|A'|\geq|A|/2$.  

If $M_1$ consists of at least two pairs, we are done. Suppose that $M_1$ consists of exactly one pair. Since $|V|\geq 8$, each player in $I$ is beaten by at least three players in $A$. (Recall that $I$ is nonempty.) One of these players is possibly used in $M_1$, and one is possibly used to match with $\mathcal{K}$, but that still leaves at least one player in $A$ that beats a player in $I$. Hence $M_1\cup M_2$ consists of at least two pairs, as desired.

\item \textit{$in_{A'}(i)\geq\log \card{V'}$ for all $i\in I$}. Let $i\in I'$. Since $M_2$ is a maximal matching, every player that contributes to the in-degree of $i$ in $A$ survives the iteration, except possibly the player that is rematched to $\mathcal{K}$. Hence the in-degree of $i$ in $A'$ is at least $\log \card{V}-1=\log(\card{V}/2)$.

\item \textit{$out(j)\leq|A'|$ for all $j\in J'$}. The condition is equivalent to $out_{B'}(j)<in_{A'}(j)$. Let $j\in J'$. We have either $in_{A'}(j)=in_A(j)$ or $in_{A'}(j)=in_A(j)-1$, where the latter case occurs exactly when a player in $A$ that beats $j$ is rematched to $\mathcal{K}$. In the former case we immediately obtain $out_{B'}(j)<in_{A'}(j)$. In the latter case, $A$ has been fully used in the two matchings before one player is rematched to $\mathcal{K}$. This means that $j$ eliminates another player in $B$, and it follows that $out_{B'}(j)\leq out_B(j)-1<in_A(j)-1=in_{A'}(j)$.

\item \textit{$\mathcal{K}$ is a king}. Let $b\in B'$. If $b\in H'$, then since $M_1$ is a maximal matching, $b$ is beaten by some player in $A'$. If $b\in I'$, then since the second invariant is maintained, $b$ is beaten by some player in $A'$. Otherwise $b\in J'$. Since the third invariant is maintained, $b$ beats at most $|A'|-1$ players in $A'$, and hence $b$ is also beaten by some player in $A'$ in this case. 

\end{enumerate}

Hence the four invariants are maintained, and the algorithm correctly computes a winning seeding for $\mathcal{K}$.
\end{proof}

Thus, we've shown a significantly general result about kings, that holds in tournaments on $n$ players, for \emph{any} power of two, answering an open research problem posed in \cite{tournijcai} to provide more general structural results that hold independent of the size of the tournament.  (Some earlier results only hold for large $n$.)

Next, we consider the weaker notion of a 3-king. \cite{kvijcai15} presented a set of conditions under which a $3$-king is an SE winner. One of their conditions is that there exists a perfect matching from the set of nodes that are reachable in exactly two steps from the $3$-king $\mathcal{K}$ onto the set of nodes that are reachable in exactly three steps from $\mathcal{K}$. Here, we present a different set of conditions that does not include the requirement of a perfect matching.

\begin{theorem}
Consider a tournament $T=(V,E)$ where $\mathcal{K}\in V$ is a 3-king. Let $A=N_{out}(\mathcal{K}),B=N_{out}(A)\cap N_{in}(\mathcal{K})$, and $C=N_{in}(\mathcal{K})\backslash B$. Suppose that the following three conditions hold:
\begin{enumerate}
\item $|A|\geq\frac{|V|}{2}$
\item $A \succ B$
\item $|B|\geq|C|$.
\end{enumerate}
Then $\mathcal{K}$ is an SE winner, and we can compute a winning seeding for $\mathcal{K}$ in polynomial time.
\end{theorem}

\begin{proof}
If $|V|=1,2,$ or $4$, the result is clear. For $|V|\geq 8$, first perform a maximal matching from $B$ to $C$ and match $\mathcal{K}$ to an arbitrary player in $A$, and then pair off players within $A$.
If $\card{A}$ is odd, then $A \cup \set{\mathcal{K}}$ matches evenly.  Else, match the remaining $a \in A$ to some $b \in B$.
We pair off players within each of $B,C$ arbitrarily, and match the remaining pair between $B$ and $C$ if needed.  After the round, $\card{A} \ge \frac{|V|}{4}$. Since the matching from $B$ to $C$ is nonempty, we still have that $|B|\geq |C|$ after the iteration. Moreover, since we applied a maximal matching, each player in $C$ is still beaten by some player in $B$.  Thus, the required conditions are maintained as invariant, and we can efficiently compute a winning seeding for $\mathcal{K}$.
\end{proof}


It would be interesting to investigate the extent to which we can weaken the (very strong) second condition that all players in $A$ beat all players in $B$.
It should be noted that if any of the three conditions is removed, the theorem no longer holds.  In particular, if the second condition is dropped, a counterexample from \cite{kvijcai15} shows that for any constant $\eps > 0$, there is a tournament on $n$ players where $\mathcal{K}$ is a $3$-king, who win against $(1-\eps)n$ players, but cannot win an SE tournament.
Given that the notion of a 3-king is significantly weaker than that of a king (recall, kings who beat $\ge \card{V}/2$ players are SE winners),
it seems reasonable to conjecture that a strong assumption such as the second condition (or the conditions seen in \cite{kvijcai15}) may be required to prove structural results for $3$-kings.

\section{SE Winners and Tournament Solutions}
\label{sec:solutions}

Tournament solutions are functions that map each tournament graph to a subset of players, usually called the choice set. The choice set is often thought of as containing the stronger players, or ``winners," within the tournament. Many tournament solutions have been considered, including the Copeland set, the Slater set, the Markov set, and the bipartisan set \cite{Laslier97,BrandtEtAlChapter}. The ability to win an SE tournament provides us with another criterion with which we can distinguish between stronger and weaker players in a tournament.  In this section, we investigate the relationship between the set of SE winners and some traditional tournament solutions.

\begin{theorem}
A player chosen by the Copeland set, the Slater set, or the Markov set is an SE winner. A player in the bipartisan set with the highest Copeland score is also an SE winner.
\label{thm:solutions}
\end{theorem}

\begin{proof}
All four tournament solutions are contained in the uncovered set.
Thus, a player from these sets will be a king, so as a special case
of Theorem~\ref{thm:generalizedstructural} (or an earlier result
of \cite{virgiaaai}), it suffices to show that the relevant
players win against at least half of the remaining players.
For the Copeland set, this is trivial
\cite{Laslier97} and \cite{LaffondLaLe93} show that players
from the Slater set and the player in the bipartisan set with the
highest Copeland score, respectively, beat at least half the players.
Next, we show that players from the Markov set win against at
least half the players.

Recall that the Markov set is defined to be the set of players of maximum probability in the stationary distribution of the Markov chain defined by the normalized Laplacian matrix $Q=(q_{ij})_{n\times n}$ of the Markov chain of the tournament, where $q_{ij} = 1/n$ if $v_i$ beats $v_j$ ($0$ otherwise) and $q_{ii} = out(v_i)/n$.  Assume that the first player is in the Markov set. It follows that the probability associated with the first player in the eigenvector $p=(p_i)_{n\times 1}$ corresponding to the eigenvalue 1 is maximal. Assume for contradiction that $q_{11}<\frac{1}{2}$. We then have
\begin{align}
p_1  &= q_{11}p_1+q_{12}p_2+\ldots+q_{1n}p_n \notag \\
     &\leq q_{11}p_1+q_{12}p_1+\ldots+q_{1n}p_1 \notag \\
     &= 2q_{11}p_1 \notag \\
     &< p_1, \notag
\end{align}
a contradiction.
\end{proof}

It is not true that any player in bipartisan set is always an SE winner. Indeed, consider a transitive tournament with the slight modification that the weakest player beats the strongest player. Then the weaker player is included in the bipartisan set even though it only beats one player and cannot be an SE winner.

Another family of tournament solutions is introduced in \cite{Laslier97} as ``iterative matrix solutions''. Consider the tournament adjacency matrix $A=(a_{ij})$, in which $a_{ij}=1$ if $i$ beats $j$, and 0 otherwise. The Copeland score is given by $A\textbf{1}$. For any positive integer $k$, we consider $A^k\textbf{1}$ and include the player(s) with the maximum resulting score in our $k$th iterative tournament solution. 

There is a natural interpretation of iterative matrix solutions as the number of paths of length $k$ starting from each player. Any player in an iterative matrix solution belongs to the uncovered set. If the player $v$ is covered by some $w$ (i.e., $w \succ \set{v} \cup N_{out}(v)$), then $v$ cannot be in the iterative matrix solution.  Indeed, if $v$ is covered by $w$, then the first steps of the paths starting from $w$ form a superset of the first steps of the paths starting from $v$. On the other hand, it is not the case that any player in an iterative matrix solution always beats at least half of the remaining players, as shown by the following example.

Consider $k=2$ and the tournament with player set $V=A\cup B\cup\{x\}$, where $A\approx rn$ and $B\approx (1-r)n$ with $\frac12<r<\frac{1}{\sqrt{3}}$. Suppose that $A\succ x\succ B\succ A$, and $A$ and $B$ are close to regular. The Copeland scores of $a\in A, b\in B, x$ are $\frac{rn}{2},\frac{(1+r)n}{2},(1-r)n$, respectively. It follows that the iterative matrix scores of $a,b,x$ are $\frac{r^2n^2}{4},\frac{(1+r^2)n^2}{4},\frac{(1-r^2)n^2}{2}$, respectively. This implies that $x$ has the maximum iterative matrix score but beats fewer than half of the remaining players.

Nevertheless, we will show that for large enough tournaments, players in an iterative matrix solution are always SE winners. First we need the following lemma and the subsequent corollary.

\begin{lemma}
\label{lemma:minkpaths}
In a tournament with $n$ players, the minimum possible number of $k$-paths is $\binom{n}{k+1}$.
\end{lemma}

\begin{proof}
In a transitive tournament, each subset of size $k+1$ gives rise to exactly one $k$-path. On the other hand, by a simple inductive argument, each subset of size $k+1$ gives rise to at least one $k$-path that goes through all $k+1$ players. The result follows immediately.
\end{proof}


\begin{corollary}
\label{cor:maxkpaths}
In a tournament with $n$ players, a player with the maximum number of $k$-paths originating from it is the origin of at least $\frac1n\binom{n}{k+1}$ $k$-paths.
\end{corollary}

We are now ready to prove the theorem.

\begin{theorem}
For any fixed $k$, there exists a constant $N_k$ such that for any tournament of size at least $N_k$, a player with the maximum number of $k$-paths originating from it is an SE winner.
\label{thm:kpaths}
\end{theorem}

\begin{proof}
Let $v$ be a player with the maximum number of $k$-paths originating from it, and let $A$ and $B$ be the sets of players who lose to $v$ and who beat $v$, respectively. From Corollary \ref{cor:maxkpaths}, $v$ is the origin of at least $\frac1n\binom{n}{k+1} \ge \frac{n^k}{2(k+1)!}$ $k$-paths for large enough $n$. Hence it must have out-degree at least $\frac{n}{2(k+1)!}$. In other words, $|A|\geq\frac{n}{2(k+1)!}$. 

If the number of players in $B$ with in-degree from $A$ less than $\log n$ is less than $|A|$, we can apply Theorem \ref{thm:generalizedstructural}. Otherwise, there are at least $|A|\geq\frac{n}{2(k+1)!}$ players in $B$ with in-degree from $A$ less than $\log n$. Call this set $H$, and consider a player $h\in H$. Since $h$ beats all but at most $\log n$ players in $A$, we can compare the number of $k$-paths originating from $v$ with the number of $k$-paths originating from $h$ by removing the common $k$-paths. The remaining number of $k$-paths originating from $v$ is at most $\log n\cdot n^{k-1}$, while by Corollary \ref{cor:maxkpaths} again, a player in $H$ with the maximum number of $k$-paths within $H$ is the origin of at least $O(n^k)$ $k$-paths, since $|H|$ is linear in $n$. This contradicts the assumption that $v$ has the maximum number of $k$-paths originating from it.
\end{proof}

\subsection{The strength of kings}
Since results concerning SE winners often involve the assumption that a player is a king in the given tournament, one might hope that there is a strong relation between SE winners and the uncovered set. For example, it could always be that a constant fraction of players in the uncovered set are SE winners, or vice versa. This is not the case, however, as the following theorem shows.

\begin{theorem}
Let $r\in(0,1)$. There exists a tournament such that the proportion of players in the uncovered set that are SE winners is less than $r$ and the proportion of SE winners that are contained in the uncovered set is also less than $r$.
\end{theorem}

\begin{proof}
Consider a tournament with player set $V=A\cup B\cup\{x,y\}$ such that 
\begin{itemize}
\item $x\succ y,B$
\item $y\succ B,A$
\item $B\succ A$
\item $A\succ x$.
\end{itemize}
The uncovered set is $A\cup\{x,y\}$.

Let $|A|=k$ and $|B|=n$. If $k<\log n$, then players in $A$ do not win enough matches to become an SE winner. Hence the proportion of players in the uncovered set that are SE winners is at most $\frac{2}{k+2}$. 

On the other hand, suppose that $B$ is a regular tournament with all players isomorphic. By symmetry, if one player in $B$ is an SE winner, then all of them are. In order for a player in $B$ to be an SE winner, players $x$ and $y$ need to be eliminated. But this can easily be done in two rounds, with $x$ beating $y$ in the first round and a player in $A$ beating $x$ in the second round. Hence the proportion of SE winners that are contained in the uncovered set is at most $\frac{2}{n+2}$.

Taking $k$ and $n$ large enough with $k<\log n$, we obtain the desired result.
\end{proof}

\section{Generative Models for Tournaments}
\label{sec:probabilistic}

Recall the Condorcet Random (CR) Model, studied in
\cite{BrM08,virgiaaai,tournwine}.  In the CR Model,
we assume there is an underlying ordering to the players,
and that, in general, stronger players win against weaker
players; however, with some small probability $p < 1/2$, the
weaker player will upset the stronger player.  In the
corresponding tournament graph, we say that for two players
$i,j$ such that $i$ occurs before $j$ in the ordering, $(i,j) \in E$ with probability $1-p$ and $(j,i) \in E$
otherwise.
A number of results are known about which players are SE winners
in tournaments drawn from a CR Model.  \cite{virgiaaai}
first showed that when $p \ge \Omega(\sqrt{\ln{n}/n})$, then
with high probability, every player in the tournament will be a
superking, and therefore an SE winner.  \cite{tournwine} shows
that even when $p \ge C \ln{n}/n$, roughly the first half of players
will be SE winners, and more generally if $p = C\cdot 2^i \ln{n}/n$,
then roughly the first $1-1/2^{i+1}$ fraction of players are SE winners.
\cite{tournwine,kvijcai15} also study various generalizations
of the CR Model.

In this section, we present improved results about tournaments
generated by the standard CR Model, showing that with high probability,
every player in a CR tournament will be an SE winner,
even with the noise $p = \Theta(\ln{n}/n)$ (with no dependence on
the player's rank).

\begin{theorem}
Let $C \ge 64$ be a constant and $p \ge C\ln{n}/n$.
Let $T$ be a tournament generated by the CR Model with noise parameter
$p$ on $n > n_C$ players (for some constant $n_C$).  With
probability $\ge 1-1/\Omega(n^2)$, every player has an efficiently-computable
winning seeding over $T$.
\label{thm:cr-all}\end{theorem}

Note that this result is asymptotically optimal, as a player
must have at least $\log{n}$ wins to be able to win an SE tournament.
If $p = o(\ln{n}/n)$, then with high probability,
the weakest player will not be able to win an SE tournament,
regardless of the seeding. The case
where $p \ge C\sqrt{\ln{n}/n}$ is covered in \cite{virgiaaai}, which
shows that every player in such a tournament is an SE winner.



The proof will use the following concentration bound, which can easily be derived from
standard Chernoff-Hoeffding bounds.

\begin{lemma}
Let $X_1, \hdots, X_n$ be independent random variables
with $X = \sum_i X_i$ and $E[X] = \mu$.
Suppose $d \le \mu$.  Then
$\Pr{X < (1-\delta)d} \le \exp(-\delta^2d/2)$.
\label{lem:varchernoff}\end{lemma}

We give a sketch of the proof before proceeding to the full proof.  First,
we argue that the weakest player $w$ will win against more than
$k \log{n}$ players in the first half, for some constant
$k$. We will think of ``swapping" $k \log{n}$ of these
losers, which we call $S$, from the first half with some arbitrary set of players
from the bottom half (so that these losers become some of the
strongest players over the second half).
Then, we argue that at least one player $v$ that $w$ beats
will be in the first $n/6$ players.
This player, with high probability, will be a king over the first
half of players, who wins against more than half the players;
thus, by \cite{virgiaaai}, this player will be an SE winner over
the first half of players.  Next, we argue that for some arbitrary
player $u$ in the weaker half of players, at least $\log{n}$ players
from the $k\log{n}$ that were swapped to the second half will beat $u$.
We then take a union bound over the players in the second half,
and argue that $w$ will be a superking over the second half, and
again by \cite{virgiaaai}, an SE winner over the second half.
Thus, $w$ will be an SE winner over the entire tournament by
winning over the weaker half, while $v$ wins against the stronger
half, and $w$ wins against $v$ in the final round.
We take a union bound over all players to arrive at the desired
result.

The detailed proof follows.

\begin{proofof}{Theorem 4.1}
Let $C \ge 64$ be a constant and $C\ln{n}/n \le p \le C\sqrt{\ln{n}/n}$.
First, note that we expect $w$ will win against $\frac{C}{2}\ln{n}
= \frac{C\ln{2}}{2}\log{n}$ players in the first half.  Next,
we can show that with high probability $w$ wins against greater than
$\frac{C\ln{2}}{4}\log{n}$ players.  Let $k = \frac{C\ln{2}}{4}$.
\begin{align*}
&\Pr{w \text{ wins against $> k\log{n}$ players in the first half }}\\
&\quad\ge 1 - \exp\left(-\frac{(k\log{n})^2}{4k\log{n}}\right)\\
&\quad= 1 - \exp\left(-\frac{k\log{n}}{4}\right)\\
&\quad= 1 - \exp\left(-\frac{C\ln{2}\log{n}}{16}\right)\\
&\quad= 1 - 1/n^{C/16}
\end{align*}

We can also argue that with probability at least $1 - 1/n^{C/6}$
$w$ wins against some player $v$ in the first $n/6$ players.
\begin{align*}
&\Pr{w \text{ wins against some } v \in [1,n/6]}\\
&\quad= 1 - (1-p)^{n/6}\\
&\quad= 1 - (1-(C\ln{n}/6)/(n/6))^{n/6}\\
&\quad\ge 1 - \exp(-C\ln{n}/6)\\
&\quad= 1 - 1/n^{C/6}
\end{align*}
where the inequality follows from the approximation
$(1-a/x)^{x} \le e^{-a}$ for $a > 0$.


In what follows, we will imagine swapping a set of $k\log{n}$ players, called $S$,
whom $w$ wins against from the first half (excluding $v$)
with $k\log{n}$ arbitrary players from the second half.
In this way, we can argue about the ``first half" and the
``second half" of players independently.
We'll argue that $v$ is an SE winner over the new ``first half"
of players, and that the inclusion of $k\log{n}$ strong players
whom $w$ beats, makes $w$ a superking over the new ``second half".

First, we argue that it is likely that $v$, whose rank is at most $n/6$,
will be an SE winner over the first
half.  In particular, with high probability,
$v$ will be a king over the first half of
players, who wins against at least $n/4$ players.
Note that we expect $v$ to win against at least
$n/3\cdot(1-p)+pn/6 - 1 = n/3 - C\ln{n}/6 -1$
players from the first half.
The out degree of $v$ is given by a random variable, which is the
sum of independent random variables, so we can bound the
probability that $out(v) < n/4$ using Lemma~\ref{lem:varchernoff}.
\begin{align*}
\Pr{out(v) \ge n/4} &\ge 1 - \exp\left(-\frac{(n/12-\frac{C}{6}\ln{n}-1)^2}
{2(n/3-\frac{C}{6}\ln{n}-1)}\right)\\
&> 1 - 1/n^{4}
\end{align*}
where the last inequality is a very loose bound on this
probability that takes effect for sufficiently large $n$.

Next, we consider the probability that $v$ is a king over the first half,
conditioned on its high out-degree.  We take a union bound
over all possible players who did not lose against $v$,
and show that it is unlikely that any of these players beat
every single player whom $v$ beat.
\begin{align*}
&\Pr{v \text{ is a king over the first half } \given out(v) \ge n/4}\\
&\quad \ge 1 - \sum_{i=1}^{n/4-1}(1-p)^{out(v)}\\
&\quad \ge 1- n/4\cdot(1-p)^{n/4}\\
&\quad \ge 1- n/4\cdot\exp(-C\ln{n}/4)\\
&\quad \ge 1- 1/4n^{C/4-1}
\end{align*}

Finally, we argue that with high probability,
$w$ will be a superking over the second half of players.
Consider some other $u$ from the second half of players.
Thus, the expected number of players from $S$
who beat $u$ is $\ge k\log{n}\cdot(1-p) = k\log{n}-\frac{kC\log^2{n}}{n}
\ge (k-1)\log{n}$ for sufficiently large $n$.
Applying Lemma~\ref{lem:varchernoff} again, we obtain the following bound.
\begin{align*}
&\Pr{u \text{ loses to fewer than } \log{n} \text{ players from } S}\\
&\quad\le \exp\left(-\frac{((k-2)\log{n})^2}{2(k-1)\log{n}}\right)\\
&\quad= \exp\left(-\frac{(k^2-4k+4)}{2(k-1)}\log{n}\right)\\
&\quad= n^{-\left(\frac{k^2-4k+4}{2\ln{2}\cdot (k-1)}\right)}
\end{align*}
Then, to guarantee that every $u$ in the second half loses to
at least $\log{n}$ players whom $w$ beats, we take a union bound
over the $n/2$ players.  For any $k > 11$, this probability will be
$\le 1/n^3$.

The overall probability that $w$ beats
a sufficiently strong king over the first half of players
is at least the following product.
\begin{align*}
&(1-1/n^{C/6})\cdot(1-1/n^{4})\cdot(1-1/4n^{C/4-1})\\
&\quad\ge1-1/n^{C/6}-1/n^{4}-1/4n^{C/4-1}\\
&\quad\ge 1 - 2/n^4
\end{align*}
Thus, the probability that some $w \in V$ wins against
$k\log{n}$ players from the (true) first half, and
wins against some $v$ strong king over the first half,
and is a superking over the second half is at least the
following.
\begin{align*}
(1-1/n^{C/16})\cdot(1-2/n^4)\cdot(1-1/n^3)
&\ge 1 - 2/n^3
\end{align*}
Taking a union bound over all players,
we conclude that with probability at least $1-1/\Omega(n^{2})$,
every player in the tournament will be an SE winner.
\end{proofof}

\subsection{Generalizing the CR Model for Tournaments}
As the prior claims demonstrate, in the standard CR Model,
every player is an SE winner with high probability, even when
upsets occur at an asymptotically minimal rate.
While this result indicates the depth of our understanding
of conditions under which a player is an SE winner,
it also suggests that the assumption that tournaments are drawn
from a CR Model -- where the noise parameter $p$ is fixed for
all matchups -- may be too rigid, incidentally providing
structure that may not exist in practical settings.
Prior work of \cite{tournwine} proposes a Generalized CR Model,
where for two players $i < j$, $j$ upsets $i$ with probability
$p \le p(i,j) \le 1/2$, for some globally specified $p$.
But even this model asserts that the probability of upsets
for \emph{every} edge must occur within the range of $[p,1/2]$.
We aim to relax our restrictions even further
in order to disrupt this structure inherent in the CR Model.

Consider the following generative model, which is parameterized
by a noise factor $p < 1/2$ and a participation factor $\Delta \le 1/2$.
The tournament on $n$ players is generated
as follows: pick any set of pairs of players $E'$ satisfying the condition that each player appears in at least $(1/2+\Delta)n$ such pairs; then, for every 
pair $\{u,v\}\in E'$, pick $(u,v)$ with probability $p_{u,v}\in [p,1-p]$, and $(v,u)$ otherwise. The probabilities $p_{u,v}$ can be arbitrary as long as they are in $[p,1-p]$.
%
The remaining edges between players may be set arbitrarily.
In this new model, many typical
arguments used in analyzing CR tournaments, including those used
in the proof of Theorem~\ref{thm:cr-all}, which hinge on the precise
definition of the CR Model, break down.

Note that unlike the CR Model, the new model does not start with an
underlying ordering of players; however, such an ordering can
easily be emulated.  For instance, to emulate the CR Model, simply
choose an ordering $\sigma$, set $\Delta = 1/2$, and for all
$u,v$ such that $\sigma(u) < \sigma(v)$, sample $(u,v)$ with
probability $1-p$.  That said, because the model does not start with an
explicit ordering, it is much more versatile.
Moreover, because only a $(1/2+\Delta)$ fraction
of the edges are determined randomly, known structures can
be (adversarially) hard-coded into the resulting graphs.
In this sense, any results that we can say about tournaments
generated from this model are extremely general and will apply
broadly.  Despite this generality, we are able to give a statement
for our model mirroring that of \cite{virgiaaai} for the CR Model.

\begin{theorem}
Let $p > c\sqrt{\frac{\log{n}}{2\Delta n}}$ for some $c > 5$.
Then with probability $> 1-\Omega(n^{(c-5)/2\ln{2}})$,
every player in a tournament $T$ sampled from the aforementioned
model has an efficiently-computable winning seeding over $T$.
\label{thm:newmodel}\end{theorem}

The proof of Theorem~\ref{thm:newmodel} is similar to the proof
of the analogous statement about the CR Model found in \cite{virgiaaai}.
It argues that with high probability every player in the tournament
will be a superking.

\begin{proofof}{Theorem~\ref{thm:newmodel}}
Let $p = c\sqrt{\frac{\log{n}}{2\Delta n}}$.  We will
argue that with high probability all nodes in a randomly
sampled tournament are superkings, so by \cite{virgiaaai}
they will be SE winners.
Let $T = (V,E)$ be a randomly sampled tournament.
We will bound the probability that $v \in V$ is not a
superking, namely, the probability that there exists
some $u \in V \setminus \set{v}$ such that $u$ loses to
fewer than $\log{n}$ players whom $v$ beats.

Let $u \in V \setminus \set{v}$.  Let $A_v$ be the
set of players $w$, for which the edge between $v$ and
$w$ was sampled randomly with probability in the range
$[p,1-p]$.  Let $A_u$ be defined analogously.  We let
$W = A_v \cap A_u$ be the players whose relation is sampled
randomly for both $v$ and $u$.
Note that we can lower bound the size of this intersection as
$\card{W} \ge (1/2+\Delta)n - 1 + (1/2+\Delta)n
- 1 - (n-2) = 2\Delta n$.
Now, note that the expected number of edges from $v$ into
$W$ is the sum of the probabilities that $(v,w)$ is an
edge for each $w \in W$, and thus is at least $2\Delta np$.
Applying Lemma~\ref{lem:varchernoff}, we can bound the
probability that this set of edges into $W$ is
smaller than $c\log{n}/p = 2\Delta np/c$.

\begin{align*}
&\Pr{\text{number of edges from $v$ into } W \le\frac{2\Delta np}{c}}\\
&\quad\le \exp\left(-(1-1/c)^2 \Delta np \right)\\
&\quad= \exp\left(-(1-1/c)^2 c\sqrt{\Delta n \log{n}/2}\right)\\
&\quad= 2^{-\Omega(\sqrt{n\log{n}})}
\end{align*}

Now, we'll condition on the fact that $v$ beats at least
$c\log{n}/p$ players from $W$.  Note that each of these
players beat $u$ with probability $\ge p$, so we expect
$\ge c\log{n}$ of these players to beat $u$.
Thus, using Lemma~\ref{lem:varchernoff} again,
we can bound the probability that $u$ does not lose to
at least $\log{n}$ of these players.
\begin{align*}
&\Pr{\text{number of edges from $W$ into } U \le \log{n}}\\
&\quad\le \exp\left(-(1-1/c)^2 c \log{n}/2\right)\\
&\quad= n^{-(1-1/c)^2c/2\ln{2}}
\end{align*}

Letting $C = (1-1/c)^2 c /2\ln{2} - 2$, by a union bound over
$v$'s opponents, the
probability that $v$ is not a superking is at most
$2^{-\Omega(\sqrt{n\log{n}})} + n^{-C-1}$.  Applying another
union bound over all players, the probability that there is
any player who is not a superking is at most $2^{\Omega(\sqrt{n\log{n}})}
+ n^{-C} \le O(n^{-C})$.  Hence with probability
$1-1/\Omega(n^{C})$, all nodes are superkings.  The result
follows from the fact that $C \ge (c-5)/2\ln{2}$.

\end{proofof}

\newpage
\bibliographystyle{alpha}
\bibliography{knockouts}

\end{document}